\documentclass[conference]{IEEEtran}

\usepackage[utf8]{inputenc}
\usepackage[T1]{fontenc}
\usepackage{amsmath,amssymb,amsthm}
\usepackage{mathtools}
\usepackage{graphicx}
\usepackage{tikz}
\usepackage[hidelinks,bookmarks=false]{hyperref}
\usepackage{xcolor}
\usepackage{cite}
\usepackage{booktabs}

\usetikzlibrary{arrows.meta,positioning,calc}

\newtheorem{lemma}{Lemma}
\newtheorem{corollary}{Corollary}
\newtheorem{proposition}{Proposition}
\theoremstyle{definition}
\newtheorem{definition}{Definition}
\newtheorem{remark}{Remark}
\newtheorem{example}{Example}

\newcommand{\Ftwo}{\mathbb{F}_2}
\newcommand{\cQ}{\mathcal{Q}}
\newcommand{\cS}{\mathcal{S}}
\newcommand{\cT}{\mathcal{T}}
\newcommand{\supp}{\operatorname{supp}}
\newcommand{\inner}[2]{\langle #1,#2\rangle}
\newcommand{\CCZ}{\mathsf{CCZ}}
\newcommand{\CZ}{\mathsf{CZ}}
\newcommand{\wt}{\operatorname{wt}}
\newcommand{\rowsp}{\operatorname{row}}

\begin{document}

\title{Native Non-Clifford Gates in Quantum LDPC Codes: Conditions, Synthesis, and Scaling Limits}

\author{%
\IEEEauthorblockN{Mohammad Rowshan\IEEEauthorrefmark{1}\IEEEauthorrefmark{2}}
\IEEEauthorblockA{\IEEEauthorrefmark{1}School of Electrical Engineering and Telecommunications, UNSW Sydney, NSW 2052, Australia}
\IEEEauthorblockA{\IEEEauthorrefmark{2}Centre for Quantum Software and Information, University of Technology Sydney, NSW 2007, Australia}
\IEEEauthorblockA{Email: mrowshan@ieee.org}
}

\maketitle

\begin{abstract}
Native constant-depth non-Clifford gates on quantum low-density parity-check (qLDPC) codes can substantially reduce the space-time overhead of magic-state distillation. This paper investigates the underlying parity and structural conditions governing parallel non-Clifford gate implementations. By analyzing their invariance, algebraic forms, and circuit synthesis, we establish fundamental limits on code scaling: strict saturated implementations cannot achieve constant-depth realization as code distance grows due to tight distance-depth bounds. We evaluate alternative scaling routes, showing that scalable designs must either relax strict subspace requirements or manage gate congestion. Through analytical bounds and code searches, we demonstrate that reconciling parallel non-Clifford operations with linear distance requires navigating these fundamental structural trade-offs.
\end{abstract}


\section{Introduction}
\label{sec:intro}

Fault-tolerant quantum computation rests on error-correcting codes that
suppress logical failure below a target rate and on the threshold theorem
guaranteeing scalability once physical noise is small
enough~\cite{aharonov1997fault,aliferis2006quantum}.
Two primary costs dominate present fault-tolerant architectures. 
For surface and color codes~\cite{kitaev2003fault,bombin2006topological},
which exhibit an asymptotic rate of 
$k/n \to 0$, the physical qubit count scales as $d^{O(1)}$ 
(e.g., $\Theta(d^2)$) with the code distance $d$, whereas the logical 
failure rate decreases exponentially with $d$. Consequently, achieving 
a target logical error rate $\epsilon$ requires a physical qubit overhead 
that scales polylogarithmically, $O\big(\mathrm{polylog}(1/\epsilon)\big)$ 
(specifically, $O(\log^2(1/\epsilon))$) with the inverse target error, 
albeit accompanied by large prefactors.
Separately, preparing non-Clifford resource states by distillation
consumes a large share of the physical
budget~\cite{bravyi2005universal,bravyi2012magic,campbell2017quantum},
even under optimized factory and compilation
designs~\cite{gidney2019efficient,litinski2019game}.

\subsection{qLDPC codes and non-Clifford gates}

Quantum LDPC codes address the first cost: bounded row and column weights
in the parity-check matrices permit constant rate together with growing
distance, with every qubit participating in $O(1)$ stabilizer
measurements~\cite{breuckmann2021quantum}.  Hypergraph-product codes~\cite{tillich2014quantum} built from suitable
classical codes attain distance $\Theta(\sqrt n)$, with rate determined by
the constituent codes and constant rate attainable.  Quantum Tanner codes and the Panteleev--Kalachev
construction achieve constant rate and linear distance with check weights
and qubit degrees bounded independently of
$n$~\cite{leverrier2022quantum,panteleev2021asymptotically}.  Hardware-oriented work has followed~\cite{bravyi2024high,pecorari2025highrate,xu2025fast}.
Locality remains a hard constraint: for geometrically local
stabilizer codes, locality imposes strong rate--distance tradeoffs; in two
dimensions the Bravyi--Poulin--Terhal bound gives $k\,d^{2}=O(n)$, with
related higher-dimensional and connectivity tradeoffs established in
subsequent
work~\cite{bravyi2010tradeoffs,baspin2022quantifying,tremblay2022constant}.

Addressing the second cost requires codes carrying non-Clifford gates
natively.  No finite-dimensional quantum error-correcting code admits a
universal set of transversal logical
gates~\cite{eastin2009restrictions}, and locality-preserving gates in
$D$-dimensional topological codes are confined to level $D$ of the
Clifford hierarchy~\cite{bravyi2013classification,pastawski2015fault}.
Individual non-Clifford gates nonetheless survive on specific families:
three-dimensional colour codes admit transversal non-Clifford gates
including $T$, while three-dimensional surface-code constructions can
support transversal non-Clifford gates such as
$\CCZ$~\cite{kubica2015universal,vasmer2019three}, whereas
two-dimensional color codes support the transversal Clifford
group~\cite{bombin2006topological}.  Recent constructions fall into two
groups.  Asymptotically good ones, with constant rate and linear distance,
include the transversal-$\CCZ$ codes of Golowich
and Guruswami~\cite{golowich2024asymptotically}, Nguyen's good binary
codes~\cite{nguyen2025good}, and the addressable constructions of He \emph{et al.}, who introduced the
addressable-orthogonality framework governing which logical coordinates a
transversal gate touches, with parameters good up to polylogarithmic
factors~\cite{he2025quantum}, and subsequently obtained asymptotically good
qubit codes with addressable transversal
$\CCZ$~\cite{he2025asymptotically}.  A second group trades linear distance for other
structure: the algebraic product codes of Golowich and
Lin~\cite{golowich2024quantum} reach dimension $K\geq N^{1-\epsilon}$ and
distance $D\geq N^{1/r}/\operatorname{poly}(\log N)$ with stabilizer weight
$\operatorname{poly}(\log N)$, so they are LDPC only in the weaker
polylogarithmic sense rather than the $O(1)$ sense used here; Lin's sheaf
codes~\cite{lin2024transversal} and Zhu's homological product
constructions~\cite{zhu2025topological}, the latter with constant rate and
polynomial distance, sit in the same group.  Parallel execution has since been
addressed directly: Gu\'emard further studies addressable transversal
multi-controlled-$Z$ gates that can be executed in parallel on good qudit
codes~\cite{guemard2025good}, Nguyen constructs good
binary codes with a transversal $\CCZ$~\cite{nguyen2025good}, and Li
\emph{et al.} recently reported transversal multi-controlled-$Z$ gates on
almost-good qLDPC and quantum locally testable
codes~\cite{li2026transversal}.  No-go
results for logical gates on product codes have also been
developed~\cite{fu2025nogo}.  Parallel logical $\CCZ$ is therefore not an
unmet target; recent constructions already provide addressable or
parallelizable transversal $\CCZ$ and related multi-controlled-$Z$ gates,
although they employ structural mechanisms different from the
representative-level synthesis studied
here~\cite{he2025quantum,he2025asymptotically,guemard2025good,li2026transversal}.

\subsection{Scope and contributions of this paper}
A pattern recurs in constructions of logical $\CCZ$: a trilinear pairing
that evaluates nontrivially on the selected logical classes while unwanted
lower-order and spectator couplings vanish.  The realizations differ widely,
over qudits, through multiplication properties of classical codes, and
through cup products on cochain complexes, and none reduces literally to a
binary coordinatewise statement.  We isolate the elementary binary case,
pairwise orthogonality of $X$-type representatives together with odd triple
overlap, and ask how much of a logical gate it delivers by itself.

Admissibility turns out not to be needed for the circuit at all.
Section~\ref{sec:local} shows that every \emph{saturated} triple, admissible
or not, admits a seven-term phase synthesis realizing the cubic phase, with
per-qubit gate incidence at most $1+3p_0+3p_0^2$ when the shared regions
have size at most $p_0$.  Admissibility enters afterwards, forcing the
shared-region parities odd and giving the circuit the intended $\CCZ$-type
conjugation relations on $X^x,X^y,X^z$.  Both are also checked
computationally.  Uniqueness of the algebraic normal form then pins the
synthesis down completely: it is the only reduced ancilla-free
$Z$-$\CZ$-$\CCZ$ circuit for that phase, its region incidences are exact,
and with $\Delta$ the largest of them and $p$ the size of the largest shared
region they yield $L\ge\Delta\ge3p\ge d$.  Inside this model a saturated
admissible triple therefore has no constant-depth realization on a family of
growing distance, independently of any packing or scheduling argument;
improving on it means leaving the model or asking for equality only on the
code space.

Passing from representatives to logical classes is where the condition
stops being automatic.  Controlling one representative at a time gives
three Schur-product tests
(Lemma~\ref{lem:one-coord}), but changing all three at once produces cross
terms those tests do not touch, and they cannot be applied in succession.
Proposition~\ref{prop:stability} gives the exact criterion, which
additionally requires $x,y,z$ and $C_X$ itself to annihilate every Schur
product of pairs of $X$-stabilizers; the latter forces the trilinear
overlap form to vanish on all of $C_X^3$, a strengthened stabilizer-subspace
form of the triorthogonality of Bravyi and Haah~\cite{bravyi2012magic}.  Section~\ref{sec:hgp-example} exhibits a
saturated admissible triple failing all three parts on a finite member of a
qLDPC family, while Section~\ref{sec:positive} gives a small code satisfying
every condition, so the criterion is neither vacuous nor automatic.  A
companion bound (Proposition~\ref{prop:pos-bound}) shows any code meeting all
the conditions has distance $O(p)$, the positive-side echo of the depth
bound.  Representative admissibility is therefore strictly weaker than a
logical gate.

Two natural scaling routes for the saturated monomial synthesis are blocked,
by counting arguments in Section~\ref{sec:obstructions}.
Saturation with bounded shared regions forces every representative in the
triple to have weight at most $3p_0$, capping the distance at $3p_0$
(Proposition~\ref{prop:constant-distance}).  Linear-weight supports with
$O(1)$ qubit participation permit only $O(1)$ candidate triples in total
(Proposition~\ref{prop:incidence}).  The saturated constant-shared-region
synthesis therefore cannot yield a constant-depth construction on a family
with growing distance, and in the linear-distance regime bounded
participation precludes a growing collection of linear-size candidate
supports.  What survives is a scheduling statement: bounded incidence on
the rank-at-most-three interaction hypergraph, rather than disjointness of
logical supports, is the property yielding constant depth
(Proposition~\ref{prop:scheduling}).

Table~\ref{tab:notation} fixes notation, and a reproducibility script
accompanies the paper.

\begin{table}[t]
  \centering
  \caption{Notation.}
  \label{tab:notation}
  \small
  \begin{tabular}{@{}lp{6.1cm}@{}}
    \toprule
    Symbol & Meaning \\
    \midrule
    $n,k,d$ & physical qubits, logical qubits, distance \\
    $H_X,H_Z$ & $X$- and $Z$-stabilizer generator matrices, $H_XH_Z^\top=0$ \\
    $C_X,C_Z$ & $\rowsp(H_X)$, $\rowsp(H_Z)$ \\
    $\mathcal{L}_X,\mathcal{L}_Z$ & $C_Z^\perp/C_X$, $C_X^\perp/C_Z$ \\
    $u\odot v$ & Schur (bitwise AND) product \\
    $D$ & $\operatorname{span}\{u\odot v: u,v\in C_X\}\supseteq C_X$ \\
    $\tau(x,y,z)$ & triple overlap, Eq.~\eqref{eq:triple} \\
    $S_t$ & combined support $\supp(x_t)\cup\supp(y_t)\cup\supp(z_t)$ \\
    $P_x,\dots,P_{xyz}$ & the seven support regions of a triple \\
    $p_0$ & bound on the four shared-region sizes \\
    $\Delta$ & maximum qubit incidence in the interaction hypergraph \\
    $L,q$ & circuit depth, maximum gate arity \\
    \bottomrule
  \end{tabular}
\end{table}

\section{CSS Conventions and $\CCZ$-Admissible Triples}
\label{sec:css}

\subsection{Conventions}

Let $H_X\in\Ftwo^{m_X\times n}$ and $H_Z\in\Ftwo^{m_Z\times n}$ have rows
specifying the $X$- and $Z$-type stabilizer generators of a CSS code
$\cQ$, subject to $H_XH_Z^\top=0$.  Put $C_X=\rowsp(H_X)$ and
$C_Z=\rowsp(H_Z)$, so $C_X\subseteq C_Z^\perp$ and
\begin{gather}
  \mathcal{L}_X=C_Z^\perp/C_X,\qquad
  \mathcal{L}_Z=C_X^\perp/C_Z,\notag\\
  k=n-\operatorname{rank}H_X-\operatorname{rank}H_Z .
\end{gather}
The distance $d$ is the least Hamming weight in
$(C_Z^\perp\setminus C_X)\cup(C_X^\perp\setminus C_Z)$.  We call $\cQ$
qLDPC when the row and column weights of $H_X$ and $H_Z$ are bounded
independently of $n$.  This convention is used throughout.

\subsection{The triple overlap and its stability}

For $x,y,z\in\Ftwo^n$ set
\begin{equation}
  \tau(x,y,z)=\sum_{i=1}^n x_iy_iz_i \bmod 2 ,
  \label{eq:triple}
\end{equation}
the parity of the number of coordinates where all three vectors are $1$.
Its operational meaning is immediate.  On three $n$-qubit registers in the
computational basis state
$|x\rangle|y\rangle|z\rangle=\bigotimes_i|x_i\rangle|y_i\rangle|z_i\rangle$,
a transversal physical $\CCZ$, one gate on the $i$th qubit of each
register, multiplies the state by $(-1)^{\tau(x,y,z)}$, since each gate
contributes $(-1)^{x_iy_iz_i}$ and no amplitude is permuted.  The synthesis of Section~\ref{sec:local} works instead inside a single
$n$-qubit register.

The form $\tau$ is multilinear over $\Ftwo$, so it need not be constant on
cosets of $C_X$.  Varying one argument gives the following.

\begin{lemma}[One-coordinate variation]
  \label{lem:one-coord}
  For $x,y,z\in\Ftwo^n$ and $u\in C_X$,
  \begin{equation}
    \tau(x+u,y,z)+\tau(x,y,z)=\inner{u}{y\odot z},
    \label{eq:rep-indep}
  \end{equation}
  where $y\odot z$ is the indicator of $\supp(y)\cap\supp(z)$.  Hence
  $\tau(x,y,z)$ is unchanged by every replacement $x\mapsto x+u$ with
  $u\in C_X$, the arguments $y$ and $z$ held fixed, if and only if
  $y\odot z\in C_X^\perp$.  The corresponding statements for the second and
  third arguments require $x\odot z\in C_X^\perp$ and
  $x\odot y\in C_X^\perp$.
\end{lemma}

\begin{proof}
  Working mod $2$, $\tau(x+u,y,z)+\tau(x,y,z)=\sum_i u_iy_iz_i
  =\inner{u}{y\odot z}$, which vanishes for all $u\in C_X$ precisely when
  $y\odot z$ is orthogonal to $C_X$.
\end{proof}

These three tests are necessary but not sufficient for $\tau$ to be a
function of the three classes.  Replacing all three representatives at once
produces cross terms such as $\tau(u,v,z)=\inner{z}{u\odot v}$ that no
one-coordinate test controls, and the tests cannot be applied in
succession: after $y\mapsto y+v$, the condition needed to vary $x$ becomes
$(y+v)\odot z\in C_X^\perp$, not $y\odot z\in C_X^\perp$.  Write
\begin{equation}
  D=\operatorname{span}_{\Ftwo}\{u\odot v:\ u,v\in C_X\},
\end{equation}
and note $C_X\subseteq D$ because $u\odot u=u$, whence
$D^\perp\subseteq C_X^\perp$.

\begin{proposition}[Stability criterion for the triple-overlap form]
  \label{prop:stability}
  For $x,y,z\in\Ftwo^n$, the equality $\tau(x+u,y+v,z+w)=\tau(x,y,z)$ holds
  for all $u,v,w\in C_X$ if and only if
  \begin{enumerate}
    \item[(S1)] $x\odot y,\ x\odot z,\ y\odot z\in C_X^\perp$;
    \item[(S2)] $x,y,z\in D^\perp$;
    \item[(S3)] $C_X\subseteq D^\perp$.
  \end{enumerate}
    Condition (S3) says the trilinear overlap form vanishes identically on
  $C_X^3$.  Because $u\odot u=u$, taking repeated arguments in (S3) forces
  $\inner{u}{v}=0$ and $\wt(u)\equiv0$ for all $u,v\in C_X$, so $C_X$ is a
  totally triorthogonal subspace: self-orthogonal, even weight, and even on
  every triple including repetitions.  This is the stabilizer-subspace
  analogue of the triorthogonality conditions of Bravyi and
  Haah~\cite{bravyi2012magic}, which are imposed on the rows of a generator
  matrix and there permit odd-weight (logical) rows that (S3) excludes.
\end{proposition}

\begin{proof}
  Multilinearity gives
  \begin{align}
    \tau(x+u,&\,y+v,z+w)+\tau(x,y,z)\notag\\
      ={}&\tau(u,y,z)+\tau(x,v,z)+\tau(x,y,w)\notag\\
       &+\tau(u,v,z)+\tau(u,y,w)+\tau(x,v,w)+\tau(u,v,w).
    \label{eq:seven-cross}
  \end{align}
  For sufficiency, note that the first three terms are $\inner{u}{y\odot
  z}$, $\inner{v}{x\odot z}$ and $\inner{w}{x\odot y}$, which vanish by
  (S1); the next three are $\inner{z}{u\odot v}$, $\inner{y}{u\odot w}$ and
  $\inner{x}{v\odot w}$, which vanish by (S2); the last is
  $\inner{w}{u\odot v}$, which vanishes by (S3).    For necessity, argue sequentially.  Setting $v=w=0$ leaves only
  $\tau(u,y,z)=\inner{u}{y\odot z}$, which must vanish for all $u\in C_X$;
  this and its two cyclic variants give (S1).  Now set $w=0$: the two terms
  linear in $u$ and $v$ vanish by (S1), leaving
  $\tau(u,v,z)=\inner{z}{u\odot v}$, which must vanish for all $u,v\in C_X$,
  hence $z\in D^\perp$, and cyclically, giving (S2).  Finally, with (S1) and
  (S2) imposed, \eqref{eq:seven-cross} reduces to $\tau(u,v,w)$, which must
  vanish for all $u,v,w\in C_X$; that is (S3).
\end{proof}

Pairwise orthogonality is subject to the same issue, since
$\inner{x+u}{y+v}=\inner{x}{y}+\inner{u}{y}+\inner{x}{v}+\inner{u}{v}$, and
$y\in C_Z^\perp$ does not imply $y\in C_X^\perp$.  It is stable exactly when
$x,y,z\in C_X^\perp$ and $C_X$ is self-orthogonal.  Both follow from (S2)
and (S3) through $D^\perp\subseteq C_X^\perp$, so
Proposition~\ref{prop:stability} already covers every clause of
Definition~\ref{def:admissible}; linear independence of the classes does
not depend on the representatives at all.

\begin{definition}[$\CCZ$-admissible triple]
  \label{def:admissible}
  Let $x,y,z\in C_Z^\perp$ be fixed vectors whose classes in
  $\mathcal{L}_X$ are linearly independent.  The triple $(x,y,z)$ is
  \emph{$\CCZ$-admissible} if
  \begin{equation}
    \inner{x}{y}=\inner{x}{z}=\inner{y}{z}=0
    \quad\text{and}\quad
    \tau(x,y,z)=1 .
  \end{equation}
  It is \emph{stably admissible} when (S1)--(S3) also hold, in which case
  every triple of representatives of the same three classes is admissible.
\end{definition}

Stability makes the parity conditions well posed on classes; it says nothing
about how a circuit built from them acts on encoded data.

A four-qubit instance shows concretely that (S1) is insufficient.  Take
$C_Z=\{0\}$ and $C_X=\operatorname{span}\{g\}$ with $g=(1,1,1,1)$, and let
\begin{equation}
  x=(0,1,1,1),\quad y=(1,0,1,1),\quad z=(1,1,0,1)
  \label{eq:s1-counterexample}
\end{equation}
as in Example~\ref{ex:4qubit}.  The three classes are independent modulo
$C_X$ and the triple is admissible.  All three Schur products are
orthogonal to $g$, since $x\odot y=(0,0,1,1)$, $x\odot z=(0,1,0,1)$ and
$y\odot z=(1,0,0,1)$ each have even weight, so (S1) holds; here $D=\{0,g\}$
and $\wt(g)=4$ is even, so (S3) holds as well.  But
$\inner{x}{g}=\inner{y}{g}=\inner{z}{g}=1$, so (S2) fails, and indeed
\begin{equation}
  \tau(x,y+g,z+g)=0\neq1=\tau(x,y,z).
\end{equation}
Changing two representatives at once therefore flips the triple overlap
even though every one-coordinate test passes.

Conditions (S1)--(S3) amount to $H_X(x\odot y)=0$ and its analogues
together with two rank computations involving $D$, all polynomial in $n$.
They are vacuous when $C_X=\{0\}$ and restrictive otherwise.  As an
implementation check, the criterion was compared against brute-force
enumeration over all $(u,v,w)\in C_X^3$ on $700$ small random instances,
agreeing on every one.

\begin{remark}[Why pairwise orthogonality]
  \label{rem:z-corrections}
    If, say, $\inner{x}{y}=1$, a diagonal layer realizing the cubic phase
  also imprints an additional representative-level quadratic phase coupling
  the first two arguments.  That phase corresponds to a $\CZ$-type factor
  only once a logical action has been identified, which is the bookkeeping
  familiar from triorthogonal and divisible-code
  constructions~\cite{bravyi2012magic}.  Requiring pairwise
  orthogonality removes those terms from the conjugation relation used in
  Proposition~\ref{prop:conjugation}.
\end{remark}

\begin{remark}[Relation to triorthogonality]
  A binary matrix is triorthogonal~\cite{bravyi2012magic} when every pair
  of distinct rows and every triple of distinct rows have even overlap.
  Odd-weight rows are permitted and serve as logical rows.  Applied to a
  whole stabilizer subspace, (S3) is stronger: it also requires vanishing
  with repeated arguments, which forces every element of $C_X$ to have
  even weight.  Definition~\ref{def:admissible} asks the opposite of its
  chosen triple, $\tau=1$, and the two requirements constrain different
  objects, (S3) the space $C_X$ and admissibility one triple of coset
  representatives.
\end{remark}

\section{Circuit Synthesis for Saturated Triples}
\label{sec:local}

\subsection{Region decomposition}

Partition $S_t=\supp(x)\cup\supp(y)\cup\supp(z)$ into seven disjoint regions
by membership: the \emph{private} $P_x,P_y,P_z$, the \emph{pairwise-shared}
$P_{xy},P_{xz},P_{yz}$, and the \emph{triple-shared} $P_{xyz}$, with
$s_{xy}=|P_{xy}|$ and so on.

\begin{lemma}[Forced parities]
  \label{lem:parities}
  If $(x,y,z)$ is $\CCZ$-admissible then $s_{xy},s_{xz},s_{yz}$ and
  $s_{xyz}$ are all odd.
\end{lemma}
\begin{proof}
  $\inner{x}{y}\equiv s_{xy}+s_{xyz}$, so $\inner{x}{y}=0$ gives
  $s_{xy}\equiv s_{xyz}$.  Also $\tau(x,y,z)\equiv s_{xyz}$, which is $1$,
  so $s_{xyz}$ is odd and hence so is $s_{xy}$; the other two follow
  identically.
\end{proof}

Call $(x,y,z)$ \emph{saturated} when $P_x=P_y=P_z=\emptyset$, and say it
has \emph{shared-region bound} $p_0$ when
$s_{xy},s_{xz},s_{yz},s_{xyz}\le p_0$.  Lemma~\ref{lem:parities} shows
that $p_0$ may as well be taken odd, since all four sizes are odd.  We
avoid the word \emph{overlap} for this bound: it constrains the four
shared regions only, and says nothing about the private ones.

\subsection{The seven-term identity}

\begin{lemma}[Diagonal synthesis]
  \label{lem:local-realization}
    Let $(x,y,z)$ be any saturated triple with shared-region bound $p_0$;
  admissibility is not needed here.  Set $\ell_x(c)=\inner{x}{c}$ and
  likewise for $y,z$.
  There is a diagonal circuit $U_t$ composed of $\CCZ$, $\CZ$ and $Z$
  gates, every leg inside $S_t$, with
  \begin{equation}
    U_t|c\rangle=(-1)^{\ell_x(c)\ell_y(c)\ell_z(c)}\,|c\rangle
    \quad\text{for all }c\in\Ftwo^n .
    \label{eq:wirewise-phase}
  \end{equation}
  In this synthesis every qubit of $S_t$ is incident to at most
  $1+3p_0+3p_0^2$ gates, and a qubit of $P_{xyz}$ attains that value when
  all four shared regions have size exactly $p_0$.
\end{lemma}

\begin{proof}
  Let $a=\sum_{i\in P_{xy}}c_i$, $b=\sum_{i\in P_{xz}}c_i$,
  $c'=\sum_{i\in P_{yz}}c_i$, $e=\sum_{i\in P_{xyz}}c_i$, all mod $2$.
  Saturation gives $\ell_x=a+b+e$, $\ell_y=a+c'+e$, $\ell_z=b+c'+e$.
  Expanding over $\Ftwo$ and reducing squares
  (Appendix~\ref{app:seven-term}),
  \begin{equation}
    \ell_x\ell_y\ell_z
      = abe+ac'e+bc'e+ae+be+c'e+e .
    \label{eq:seven-terms}
  \end{equation}
  Substituting $a=\sum_{P_{xy}}c_i$ and so on expands each term into
  monomials in the $c_i$ of degree equal to the number of region symbols
  present.  Realize each cubic monomial by one $\CCZ$, each quadratic
  monomial by one $\CZ$, and the linear term $e$ by one $Z$ on each qubit
  of $P_{xyz}$.  Saturation places every index inside $S_t$, so
  \eqref{eq:wirewise-phase} holds by construction.  For the incidence
  count, a qubit of $P_{xyz}$ occurs in all seven terms: three cubic terms
  pair it with two shared regions each, at most $p_0^2$ monomials apiece;
  three quadratic terms pair it with one region each, at most $p_0$
  apiece; the linear term contributes one.  A qubit of $P_{xy}$ occurs
  only in the terms containing $a$, giving at most $2p_0^2+p_0$.  Both are
  bounded by $1+3p_0+3p_0^2$, with equality at a qubit of $P_{xyz}$
  whenever all four shared regions have size exactly $p_0$.  For an
  \emph{admissible} triple Lemma~\ref{lem:parities} restricts such
  configurations to odd $p_0$; without admissibility they occur at even $p_0$
  as well, as the script exhibits.
\end{proof}

Within the circuit model of \eqref{eq:wirewise-phase} the count is not
merely an upper bound from one synthesis: it is forced.

\begin{lemma}[Uniqueness of the reduced ancilla-free diagonal gate multiset]
  \label{lem:anf}
  Let $\mathcal{V}$ be the set of ancilla-free circuits on $n$ qubits built
  from gates $Z_i$, $\CZ_{ij}$ and $\CCZ_{ijk}$.  Every $V\in\mathcal{V}$ is
  diagonal with $V|c\rangle=(-1)^{P_V(c)}|c\rangle$, where
  \begin{equation}
    P_V(c)=\sum_i\alpha_ic_i+\sum_{i<j}\beta_{ij}c_ic_j
           +\sum_{i<j<k}\gamma_{ijk}c_ic_jc_k
    \label{eq:anf}
  \end{equation}
  and $\alpha,\beta,\gamma$ are the parities of the multiplicities of the
  corresponding gates in $V$.  Conversely, let $Q:\Ftwo^n\to\Ftwo$ have
  degree at most three with $Q(0)=0$.  Then some $V\in\mathcal{V}$ realizes
  $Q$ on every $c\in\Ftwo^n$, and after cancelling repeated gates its gate
  multiset is unique: it contains $Z_i$ exactly when $\alpha_i=1$,
  $\CZ_{ij}$ exactly when $\beta_{ij}=1$ and $\CCZ_{ijk}$ exactly when
  $\gamma_{ijk}=1$, for the coefficients of the algebraic normal form of
  $Q$.  The gates commute, so the ordering and layering of that multiset
  remain free.
\end{lemma}

\begin{proof}
  Each listed gate is diagonal, Hermitian and involutory, and any two of them
  commute; a gate on leg set $T$ contributes the phase $\prod_{i\in T}c_i$.
  Hence $P_V=\sum_{\text{gates}}\prod_{i\in T}c_i$ over $\Ftwo$, which is
  \eqref{eq:anf} with the stated coefficients, and repeated gates cancel in
  pairs.  The reachable functions are exactly those of degree at most three
  with no constant term, since no gate contributes a constant.  The algebraic
  normal form of a Boolean function is unique, so $P_V=Q$ determines
  $\alpha,\beta,\gamma$ and therefore the reduced gate multiset.
\end{proof}

The uniqueness has a consequence that needs neither equal region sizes nor
any packing argument.

\begin{corollary}[A distance lower bound on exact-synthesis depth]
  \label{cor:exact}
  Let $(x,y,z)$ be a saturated $\CCZ$-admissible triple, write
  $a=s_{xy}$, $b=s_{xz}$, $c=s_{yz}$, $e=s_{xyz}$, and put
  $p=\max\{a,b,c,e\}$.  In the reduced gate multiset of
  Lemma~\ref{lem:anf} a qubit of each region is incident to exactly
  \begin{equation}
    \begin{aligned}
      \delta_{xy}&=e(b+c+1), &\delta_{xz}&=e(a+c+1),\\
      \delta_{yz}&=e(a+b+1), &\delta_{xyz}&=ab+ac+bc\\
      &&&\quad{}+a+b+c+1
    \end{aligned}
    \label{eq:incidences}
  \end{equation}
  gates.  Its maximum incidence $\Delta$ and any depth $L$ realizing it
  therefore satisfy
  \begin{equation}
    L\;\ge\;\Delta\;\ge\;3p\;\ge\;d .
    \label{eq:depth-chain}
  \end{equation}
  Consequently no such realization has constant depth on a family with
  $d_n\to\infty$.
\end{corollary}

\begin{proof}
  The seven groups of monomials produced by $ABE,ACE,BCE,AE,BE,CE,E$ are
  pairwise disjoint, because the four regions are disjoint and each group
  draws one leg from a fixed set of them.  Counting the monomials containing
  a fixed qubit gives \eqref{eq:incidences}: a qubit of $P_{xy}$ occurs in
  $ABE$ ($be$ monomials), $ACE$ ($ce$) and $AE$ ($e$), and cyclically, while
  a qubit of $P_{xyz}$ occurs in all seven groups, contributing
  $ab+ac+bc$ from the cubic ones, $a+b+c$ from the quadratic ones and $1$
  from $E$.    By Lemma~\ref{lem:parities}, $a,b,c,e$ are odd, hence at least
  $1$.  If $p=e$ then $\delta_{xy}=e(b+c+1)\ge3e=3p$; if $p=a$ then
  $\delta_{xyz}\ge ab+ac+a=a(b+c+1)\ge3a=3p$, and cyclically for $b$ and
  $c$; so $\Delta\ge3p$ in every case.  For the depth, note that any
  $V\in\mathcal{V}$ realizing the phase has, by Lemma~\ref{lem:anf}, gate
  multiplicities whose parities are the algebraic-normal-form coefficients,
  so each gate of the reduced multiset occurs an odd and therefore positive
  number of times in $V$.  A qubit incident to $\Delta$ reduced gates is
  thus incident to at least $\Delta$ gate instances of $V$, and those
  pairwise share it and so occupy distinct layers, giving $L\ge\Delta$.
  Finally $x$ represents a nontrivial $X$-logical class, so
  $d\le\wt(x)=a+b+e\le3p$.
\end{proof}

Lowering the incidence therefore requires either leaving $\mathcal{V}$, by
introducing ancillas with non-diagonal gates such as CNOT or enlarging the
gate set, or weakening the requirement from equality on all of $\Ftwo^n$ to
equality only on $C_Z^\perp$ or on the code space.  The second alternative
stays inside $\mathcal{V}$: it realizes a different Boolean function that
happens to agree with the target where it matters.

For $p_0=1$ the count is seven, matching
$\CCZ_{1,2,4}\CCZ_{1,3,4}\CCZ_{2,3,4}\CZ_{1,4}\CZ_{2,4}\CZ_{3,4}Z_4$ for the
triple of Example~\ref{ex:4qubit}.  That case and the $p_0=3$ case of
Section~\ref{sec:hgp-example} were checked by exhaustive evaluation of
\eqref{eq:wirewise-phase} over all basis states of $S_t$, and the script
confirms that the synthesized gate list coincides with the support of the
algebraic normal form.

\begin{proposition}[Pauli conjugation]
  \label{prop:conjugation}
  Let $(x,y,z)$ be a saturated \emph{and $\CCZ$-admissible} triple and let
  $U_t$ be the circuit of Lemma~\ref{lem:local-realization}.  Then $\wt(h)$
  is odd for each $h\in\{x,y,z\}$, and
  \begin{equation}
    U_t\,X^{x}\,U_t^\dagger=X^{x}D_{yz},\qquad
    D_{yz}=\mathrm{diag}\bigl((-1)^{\ell_y(c)\ell_z(c)}\bigr),
  \end{equation}
  with the cyclic analogues for $X^y$ and $X^z$.
\end{proposition}

\begin{proof}
  By Lemma~\ref{lem:parities}, $\wt(x)=s_{xy}+s_{xz}+s_{xyz}$ is a sum of
  three odd numbers, hence odd.  $U_t$ is a product of commuting Hermitian
  involutions, so $U_t^\dagger=U_t$ and $U_t|c\rangle=(-1)^{Q(c)}|c\rangle$
  with $Q=\ell_x\ell_y\ell_z$.  Then $U_tX^xU_t|c\rangle
  =(-1)^{Q(c)+Q(c+x)}|c+x\rangle$.  Pairwise orthogonality gives
  $\ell_y(c+x)=\ell_y(c)$ and $\ell_z(c+x)=\ell_z(c)$, while
  $\ell_x(c+x)=\ell_x(c)+\wt(x)$, so
    $Q(c)+Q(c+x)=\wt(x)\ell_y(c)\ell_z(c)=\ell_y(c)\ell_z(c)$.
\end{proof}

Both assumptions enter.  Without admissibility, $x=y=z=(1,1)$ is saturated
with even weight.  Without saturation, adjoining a private coordinate to $x$
in Example~\ref{ex:4qubit} leaves the triple admissible with odd
shared-region parities but $\wt(x)=a+b+e+|P_x|=4$ even, and the identity then
fails wherever $\ell_y\ell_z=1$.  Pairwise orthogonality is needed on its own:
$x=y=z=(1)$ gives $U_t=Z$ and $ZXZ^\dagger=-X$, not $XZ$.  The script checks
all three.

This has the algebraic form of the $\CCZ$ conjugation rule, with $D_{yz}$ a
representative-level quadratic phase; it becomes an encoded $\CCZ$ relation
only once the linear forms $\ell_x,\ell_y,\ell_z$ are identified with logical
coordinates.  Code-space preservation, trivial action on the other logical
coordinates, and that identification are gathered as Hypothesis~(H3) in
Section~\ref{sec:main}, where Lemma~\ref{lem:codespace} settles the first.

\subsection{Unsaturated triples}
\label{sec:unsaturated}

Dropping saturation changes \eqref{eq:seven-terms} substantially.  With
$\beta_x=\sum_{P_x}c_i$ and so on, the Zhegalkin expansion of
$\ell_x\ell_y\ell_z$ has, as a full symbolic expansion in the seven region-parity variables,
$38$ surviving monomials rather than seven, the $31$ additional ones each
involving a private region; particular terms drop out when the
corresponding regions are empty.  Among them is $\beta_x\beta_y\beta_z$, which
expands into $|P_x||P_y||P_z|$ cubic monomials.  By
Lemma~\ref{lem:anf} those monomials are present in the algebraic normal form
of $\ell_x\ell_y\ell_z$ and so appear in \emph{every} circuit of
$\mathcal{V}$ realizing the phase, whence a qubit of $P_x$ is incident to at
least $|P_y||P_z|$ gates.  This is a lower bound rather than a defect of one
expansion: bounding the four shared regions cannot bound gate incidence once
the private regions grow, for any ancilla-free $Z$-$\CZ$-$\CCZ$ circuit.
Computing each $\ell_h$ onto an ancilla by a fan-in tree, applying one
$\CCZ$ and uncomputing leaves $\mathcal{V}$ and avoids the blowup in gate
count, at the price of depth $O(\log|S_t|)$.  Open Problem~2 asks what an
ancilla-assisted route can achieve.

\section{Two Scaling Obstructions}
\label{sec:obstructions}

Corollary~\ref{cor:exact} already rules out constant depth for a single
saturated admissible triple inside the ancilla-free model.  The two counting
arguments below are independent of the circuit model: they concern how many
candidate triples a code family can supply at all, and they close off the
natural strategy of packing many triples with disjoint supports into one
parallel layer.

\begin{proposition}[Saturation caps the distance]
  \label{prop:constant-distance}
  Let $(x,y,z)$ be saturated with shared-region bound $p_0$.  Then
  \begin{equation}
    |S_t|\le 4p_0,\qquad \wt(x),\wt(y),\wt(z)\le 3p_0 .
  \end{equation}
  If in addition the classes of $x,y,z$ are nontrivial in
  $\mathcal{L}_X$, the code satisfies $d\le 3p_0$.  Consequently a family
  $\{\cQ_n\}$ with $d_n\to\infty$ admits no saturated admissible triples
  of constant shared-region bound, and one with $d_n=\Theta(n)$ admits
  none with $p_0=o(n)$.
\end{proposition}

\begin{proof}
  Saturation gives $S_t=P_{xy}\sqcup P_{xz}\sqcup P_{yz}\sqcup P_{xyz}$,
  so $|S_t|=s_{xy}+s_{xz}+s_{yz}+s_{xyz}\le4p_0$, and
  $\wt(x)=s_{xy}+s_{xz}+s_{xyz}\le3p_0$.  A nontrivial class in
  $\mathcal{L}_X$ has a representative of weight at most $3p_0$, and the
  distance is the minimum weight over all such representatives, whence
  $d\le3p_0$.
\end{proof}

\begin{proposition}[Incidence bound on candidate collections]
  \label{prop:incidence}
  Let $\cS$ be a collection of triples with combined supports satisfying
  $|S_t|\ge\sigma$ for all $t$, and suppose every qubit lies in at most
  $M$ of the supports.    Then $|\cS|\le Mn/\sigma$.  In particular, if $M=O(1)$ and $\sigma=an$
  then $|\cS|\le M/a=O(1)$; more generally $M=O(1)$ with $\sigma\ge1$ gives
  $|\cS|=O(n)$, so no collection of size $\Omega(n^{1+\gamma})$ with
  $\gamma>0$ satisfies both conditions.
\end{proposition}

\begin{proof}
  Count incidences $I=|\{(i,t):i\in S_t\}|$ two ways:
  $I=\sum_{t}|S_t|\ge\sigma|\cS|$ and $I=\sum_i|\{t:i\in S_t\}|\le Mn$.
\end{proof}

Together the propositions rule out the bounded-region saturated route and
the bounded-participation linear-support route.  By
Proposition~\ref{prop:constant-distance}, the saturated bounded-region
synthesis has $|S_t|=O(1)$ and can therefore occur only on codes with
$d=O(1)$; by Proposition~\ref{prop:incidence}, linear-weight supports,
which linear distance requires, admit only $O(1)$ candidates once
participation is bounded.  Within the disjoint-support strategy considered
here, either the saturated supports are small, in which case the code has
bounded distance, or linear-size supports permit only constantly many
candidates under bounded participation.  Obstructions of a different kind, based on the topology of product codes
rather than on counting, appear in~\cite{fu2025nogo}.  A greedy packing lemma
remains valid in the small-support regime and is recorded next for
completeness, but by Proposition~\ref{prop:constant-distance} it applies
only to families of bounded distance.  This leaves open the unsaturated
bounded-congestion constructions posed in Section~\ref{sec:implications}.

\begin{lemma}[Packing, small-support regime]
  \label{lem:packing-main}
  Let $\{S_t\}_{t\in\cS}$ satisfy $|S_t|\le s_0$ and let every qubit lie
  in at most $M$ supports.  Then some $\cT\subseteq\cS$ has pairwise
  disjoint supports with $|\cT|\ge|\cS|/(Ms_0)$.
\end{lemma}

\begin{proof}
  Greedily select $t$, add it to $\cT$, and discard every $t'$ with
  $S_{t'}\cap S_t\ne\emptyset$.  Each selection discards at most
  $M|S_t|\le Ms_0$ elements, so $|\cS|\le|\cT|Ms_0$.
\end{proof}

Support disjointness is stronger than the scheduling argument requires.
What the coloring argument of the next section consumes is bounded
incidence of qubits in \emph{gates}, which candidate triples with heavily
overlapping supports can still satisfy.

\section{Parallel Logical $\CCZ$ from Bounded Congestion}
\label{sec:main}

\subsection{Interaction hypergraphs}

A diagonal circuit built from $Z$, $\CZ$ and $\CCZ$ gates is described by
a hypergraph $H=(V,E)$ with $V=[n]$ and $E$ a multiset of edges of rank at
most three, one edge per gate.  The incidence of a qubit is
$\deg_H(v)=|\{e\in E: v\in e\}|$ and $\Delta(H)=\max_v\deg_H(v)$.  These gates all commute; under the standard circuit-layer model two may
share a layer when their supports are disjoint.  Rank at most three, rather
than exact $3$-uniformity, is what
Lemma~\ref{lem:local-realization} produces, and the coloring bound is
insensitive to the difference.

\begin{lemma}[Coloring and depth]
  \label{lem:coloring-main}
  A hypergraph $H$ of rank at most $q$ with $\Delta(H)\le\Delta$ admits a
  proper edge coloring with at most $q(\Delta-1)+1$ colors, hence its
  gates execute in a circuit of depth at most $q(\Delta-1)+1$.
\end{lemma}

\begin{proof}
  Color greedily in any order.  An edge $e$ meets at most $\Delta-1$ other
  edges through each of its at most $q$ vertices, so at most $q(\Delta-1)$
  colors are excluded when $e$ is reached.  Each color class consists of
  pairwise disjoint edges and forms one parallel layer.
\end{proof}

For $q=3$ this gives $3\Delta-2$ colors.  Every $Z$, $\CZ$ and $\CCZ$ gate
already appears as an edge of rank one, two or three in the same
hypergraph, so no separate layer for the lower-rank terms is needed.  The
bound is worst case; under additional near-regularity and small-codegree
hypotheses, asymptotically sharper edge-colouring bounds are
known~\cite{pippenger1989asymptotic}, and would reduce the constant rather
than the scaling.

\subsection{Scheduling proposition}

\begin{proposition}[Parallel scheduling under bounded gate congestion]
  \label{prop:scheduling}
    Let $\cQ$ be a CSS code on $n$ qubits with distance $d$, and let $\cT$
  index a finite collection of triples of logical qubits.  For each
  $t\in\cT$ choose physical representatives $x_t,y_t,z_t\in C_Z^\perp$ of
  the three classes, put
  $S_t=\supp(x_t)\cup\supp(y_t)\cup\supp(z_t)$, and let $U_t$ be an
  associated physical diagonal circuit.  Assume:
  \begin{itemize}
        \item[(H1)] \emph{Realizability.}  Each $U_t$ is built from $Z$, $\CZ$
      and $\CCZ$ gates supported on $S_t$ and satisfies
      \eqref{eq:wirewise-phase}.  This ties the statement to
      Section~\ref{sec:local}; once the circuits $U_t$ are given, only
      (H2)--(H4) enter the argument below.
    \item[(H2)] \emph{Bounded congestion.}  The union $H=(V,E)$ of the
      gate edges of $\{U_t\}_{t\in\cT}$, a hypergraph of rank at most
      three, has $\Delta(H)\le\Delta$ for a constant $\Delta$.
    \item[(H3)] \emph{Logical action.}  Each $U_t$ preserves the code
      space of $\cQ$, acts as the logical $\CCZ$ on the three logical
      qubits addressed by the classes of $x_t,y_t,z_t$ in a fixed
      symplectic basis, and acts trivially on the remaining logical
      qubits.
        \item[(H4)] \emph{Disjoint logical addressing.}  In the fixed logical
      basis of (H3), the three classes assigned to each $t$ are three
      distinct logical-$X$ basis elements, and the basis elements assigned
      to distinct $t$ are disjoint.  This implies linear independence of
      the $3|\cT|$ classes involved, and is strictly stronger.
  \end{itemize}
    Then the gates of $\{U_t\}$ execute in a circuit $U$ of depth
  $L\le3\Delta-2$ implementing $\bigotimes_{t\in\cT}\CCZ_t$ on the
    addressed logical qubits.  Moreover, for a fault occurring at one time
  step of $U$ and supported on $w$ qubits, the propagated fault operator can
  be supported on at most $3^Lw$ output qubits.
\end{proposition}

\begin{proof}
    (H2) and Lemma~\ref{lem:coloring-main} with $q=3$ give a proper edge
  coloring in at most $3\Delta-2$ colors, and executing one color class per
  time step yields a circuit of that depth; the rank-one and rank-two
  edges are coloured alongside the rank-three ones and need no extra
  layer.  By (H3) each $U_t$ implements the logical
  $\CCZ$ on its own three logical qubits and the identity elsewhere, and
  by (H4) these triples of logical qubits are pairwise distinct, so the
  logical actions compose to $\bigotimes_t\CCZ_t$ regardless of the order
  in which color classes are applied, the gates being diagonal and hence
  commuting.  The support statement is Lemma~\ref{lem:light-cone} with
  arity $q=3$.
\end{proof}

\begin{remark}[What is assumed]
  \label{rem:assumptions}
      (H1) always holds: $\ell_x\ell_y\ell_z$ has degree at most three, so its
  algebraic normal form is realizable by $Z$, $\CZ$ and $\CCZ$ gates for any
  triple, and by Lemma~\ref{lem:anf} that realization is the only reduced one
  inside the model.  For a fixed triple there is therefore no choice to make;
  the question inside the model is whether a code family supplies triples
  whose ANF realizations collectively have bounded congestion, and
  Corollary~\ref{cor:exact} answers it negatively for saturated admissible
  triples on growing-distance families.  Lowering the congestion of a fixed
  phase requires either leaving the model, through ancillas with non-diagonal
  gates or a larger gate set, or asking only for equality on the code space,
  which stays inside it.    (H2)
  replaces the disjoint-support hypothesis ruled out in
  Section~\ref{sec:obstructions}: it constrains the physical gate schedule,
  not logical supports, so triples with heavily overlapping supports may
  still meet it.  (H3) is the substantive gap and does not follow from
  Definition~\ref{def:admissible}, since
  Proposition~\ref{prop:conjugation} verifies only the conjugation of
  $X^{x_t},X^{y_t},X^{z_t}$, while triviality on spectator logical qubits
  needs the global structure of addressable-orthogonality
  frameworks~\cite{he2025quantum}.
\end{remark}

Part of (H3) is decidable directly.  Since every $U_t$ here is diagonal, the
following gives an exact checkable criterion for code-space preservation,
applied to a concrete code in Section~\ref{sec:hgp-example}.

\begin{lemma}[Code-space preservation for diagonal circuits]
  \label{lem:codespace}
  Let $U$ be diagonal with $U|c\rangle=(-1)^{Q(c)}|c\rangle$ for a function
  $Q:\Ftwo^n\to\Ftwo$.  Then $U$ maps the code space of
  $\mathrm{CSS}(C_X,C_Z)$ into itself if and only if $Q$ is constant on
  every coset $c+C_X$ with $c\in C_Z^\perp$.  In particular, if some
  $u\in C_X$ has $Q(u)\neq Q(0)$, then $U$ does not preserve the code space.
\end{lemma}

Stability already delivers this.  Condition (S2) places $x,y,z$ in $D^\perp$,
and $C_X\subseteq D$ because $u\odot u=u$, so $x,y,z\in C_X^\perp$; then
$\ell_x(c+u)=\ell_x(c)+\inner{x}{u}=\ell_x(c)$ for every $u\in C_X$, and
likewise for $y$ and $z$, so $Q=\ell_x\ell_y\ell_z$ is constant on each
coset.  A stably admissible triple therefore satisfies the first clause of
(H3) automatically, and only the identification of the induced logical phase
remains open.

\begin{proof}
  A logical basis state is $|\bar c\rangle=|C_X|^{-1/2}\sum_{u\in C_X}
  |c+u\rangle$ for $c\in C_Z^\perp$, one per coset of $C_X$ in $C_Z^\perp$.
  Then $U|\bar c\rangle=|C_X|^{-1/2}\sum_u(-1)^{Q(c+u)}|c+u\rangle$, and
  applying $X^{u'}$ for $u'\in C_X$ permutes the summands, returning the
  same state exactly when $Q(c+u+u')=Q(c+u)$ for all $u$.  So $U|\bar
  c\rangle$ is fixed by every $X$-stabilizer precisely when $Q$ is constant
  on $c+C_X$; it remains a $+1$ eigenstate of every $Z$-type stabilizer because each
  basis vector of $c+C_X$ lies in $C_Z^\perp$ and a diagonal circuit does not
  change computational-basis support.  Since $U$ is diagonal
  it maps distinct cosets to disjointly supported vectors, so the condition
  is needed for each $c$ separately.  Taking $c=0\in C_Z^\perp$ gives the
  last claim.
\end{proof}

\begin{remark}[Error spread, not distance preservation]
  \label{rem:distance}
        Hypothesis (H3) makes each $U_t$, hence $U$, map the code space to itself,
  so $U\cQ=\cQ$ is the same stabilizer code with the same distance $d$.
  Write a single fault as $E$ inserted between two segments of the circuit,
  $U=U_2U_1$.  Then the faulty output is
  $U_2EU_1|\psi\rangle=(U_2EU_2^\dagger)\,U|\psi\rangle$, so the deviation
  operator $U_2EU_2^\dagger$ acts on the ideal output $U|\psi\rangle\in\cQ$
  and, by Lemma~\ref{lem:light-cone}, is supported on at most $t=3^Lw$
  qubits.  Ideal recovery then corrects it whenever $2t<d$, and this holds
  for an arbitrary operator on those qubits, not merely a Pauli one, since
  any such operator expands in the Pauli basis supported there; that $U$ is
  non-Clifford is immaterial.  What the light-cone bound does not deliver is
  circuit-level fault tolerance: it says nothing about several faults at
  different spacetime locations, faulty syndrome extraction, or the recovery
    operation itself.  We use the term \emph{light-cone protection} for what is
  proved and avoid \emph{distance preservation}.
\end{remark}

The proposition isolates the scheduling consequence of bounded congestion;
its substantive assumption is (H3).  Fig.~\ref{fig:scheduling} illustrates
the coloring.

\begin{figure}[t]
  \centering
  \begin{tikzpicture}[scale=0.92,
    qubit/.style={circle,draw,minimum size=0.8em,inner sep=0pt,font=\scriptsize},
    hub/.style={circle,fill=black,inner sep=0.7pt},
    edgeA/.style={thick},
    edgeB/.style={thick,dashed},
    edgeC/.style={thick,dotted}
  ]
    \foreach \i in {1,...,9} {\node[qubit] (q\i) at (0.7*\i,0) {\i};}
    \node[hub] (h1) at (1.63,0.9) {};
    \draw[edgeA] (h1) -- (q1); \draw[edgeA] (h1) -- (q2); \draw[edgeA] (h1) -- (q4);
    \node[hub] (h4) at (4.90,0.9) {};
    \draw[edgeA] (h4) -- (q5); \draw[edgeA] (h4) -- (q7); \draw[edgeA] (h4) -- (q9);
    \node[hub] (h2) at (2.33,1.5) {};
    \draw[edgeB] (h2) -- (q2); \draw[edgeB] (h2) -- (q3); \draw[edgeB] (h2) -- (q5);
    \node[hub] (h3) at (3.97,1.5) {};
    \draw[edgeB] (h3) -- (q4); \draw[edgeB] (h3) -- (q6); \draw[edgeB] (h3) -- (q7);
    \node[hub] (h5) at (3.97,2.1) {};
    \draw[edgeC] (h5) -- (q2); \draw[edgeC] (h5) -- (q6); \draw[edgeC] (h5) -- (q9);
    \node[anchor=west,font=\scriptsize] at (6.6,0.9) {Layer 1};
    \node[anchor=west,font=\scriptsize] at (6.6,1.5) {Layer 2};
    \node[anchor=west,font=\scriptsize] at (6.6,2.1) {Layer 3};
  \end{tikzpicture}
  \caption{Edge coloring of five $\CCZ$ gates on nine qubits.  Each gate is
    drawn as a fan-out from a hub to its three legs.  Qubit $2$ has
        incidence $3$, so Lemma~\ref{lem:coloring-main} guarantees at most
    $3\Delta-2=7$ layers; three suffice here, the bound being worst case.}
  \label{fig:scheduling}
\end{figure}

\section{Numerical Results}
\label{sec:examples}

All numerical values in this section are generated by the accompanying software.

\begin{example}[Four-qubit representative triple]
  \label{ex:4qubit}
  Take $C_X=C_Z=\{0\}$ on $n=4$ qubits, so every nonzero vector represents
  a nontrivial $X$-logical class and Proposition~\ref{prop:stability} holds
  vacuously.  With
  \begin{equation}
    x=(0,1,1,1),\quad y=(1,0,1,1),\quad z=(1,1,0,1),
  \end{equation}
    the three pairwise inner products vanish and $\tau(x,y,z)=1$, so the
  triple is $\CCZ$-admissible.  It is saturated with
  $s_{xy}=s_{xz}=s_{yz}=s_{xyz}=1$, giving $p_0=1$ and combined support
  $S_t=\{1,2,3,4\}$ of size $4$.  Lemma~\ref{lem:local-realization}
  synthesizes seven gates, three $\CCZ$, three $\CZ$ and one $Z$, with
  qubit $4$ incident to all seven, matching $1+3p_0+3p_0^2=7$.
  Fig.~\ref{fig:triple} shows the supports.

  Here the logical action can be pinned down outright, which the
  larger examples do not permit.  Adjoin $r=(1,1,1,0)$.  A direct
  computation gives $\inner{x}{x}=\inner{y}{y}=\inner{z}{z}=\inner{r}{r}=1$
  with all six cross products zero, so $\{x,y,z,r\}$ is a self-dual basis
  of $\Ftwo^4$.  Coordinates in that basis are exactly
  $\bigl(\ell_x(c),\ell_y(c),\ell_z(c),\ell_r(c)\bigr)$, and since
  $C_X=\{0\}$ collapses physical and logical bits, the phase
  $(-1)^{\ell_x(c)\ell_y(c)\ell_z(c)}$ of \eqref{eq:wirewise-phase} is
  logical $\CCZ$ on the first three logical qubits and the identity on the
  fourth.  Hypotheses (H3) and (H4) of Proposition~\ref{prop:scheduling} therefore
  hold for this triple, which is what makes
  Examples~\ref{ex:disjoint} and~\ref{ex:scaling} statements about logical
  operations rather than about phases alone.
\end{example}

\begin{figure}[t]
  \centering
  \begin{tikzpicture}[scale=0.9,
    qubit/.style={circle,draw,minimum size=0.9em,inner sep=0pt},
    xmark/.style={very thick},
    ymark/.style={very thick,dashed},
    zmark/.style={very thick,dotted},
  ]
    \node[qubit] (q1) at (0,0)   {1};
    \node[qubit] (q2) at (1.8,0) {2};
    \node[qubit] (q3) at (3.6,0) {3};
    \node[qubit] (q4) at (5.4,0) {4};
    \draw[xmark] (q2) circle [radius=0.45];
    \draw[xmark] (q3) circle [radius=0.45];
    \draw[xmark] (q4) circle [radius=0.45];
    \draw[ymark] (q1) circle [radius=0.60];
    \draw[ymark] (q3) circle [radius=0.60];
    \draw[ymark] (q4) circle [radius=0.60];
    \draw[zmark] (q1) circle [radius=0.75];
    \draw[zmark] (q2) circle [radius=0.75];
    \draw[zmark] (q4) circle [radius=0.75];
    \draw[xmark] (-0.3,1.3) -- ++(0.4,0);
    \node[anchor=west] at (0.2,1.3) {$x$};
    \draw[ymark] (1.5,1.3) -- ++(0.4,0);
    \node[anchor=west] at (2.0,1.3) {$y$};
    \draw[zmark] (3.0,1.3) -- ++(0.4,0);
    \node[anchor=west] at (3.5,1.3) {$z$};
  \end{tikzpicture}
  \caption{Supports of the triple of Example~\ref{ex:4qubit}.  Solid,
    dashed and dotted rings mark $\supp(x)$, $\supp(y)$, $\supp(z)$.
    Qubit $4$ is the triple-shared region; qubits $1,2,3$ are the three
    pairwise-shared regions; no private region is nonempty.}
  \label{fig:triple}
\end{figure}

\begin{example}[Two blocks at no additional depth]
  \label{ex:disjoint}
  Take two disjoint copies of Example~\ref{ex:4qubit}, on qubits
  $\{1,\dots,4\}$ and $\{5,\dots,8\}$, giving $14$ gates.  Within each
  copy all seven synthesized gates touch the triple-shared qubit, qubit
  $4$ and qubit $8$ respectively, so this monomial synthesis has depth
  seven.  The two copies are physically disjoint, so corresponding gates
  from both schedule into the same layers, and two copies need the same
  depth seven as one.  Depth here responds to congestion, not to gate
  count, which is what (H2) tracks.
\end{example}

\begin{example}[Block stacking and the distance cap]
  \label{ex:scaling}
  Stacking $m$ disjoint copies of Example~\ref{ex:4qubit} gives $n=4m$
  qubits carrying $m$ pairwise disjoint saturated triples, hence $m$
  parallel logical operations at constant depth, with
  $|\cT|=\Theta(n)$.  The construction satisfies
  Lemma~\ref{lem:packing-main} with $s_0=4$, $M=1$.  It does not
  contradict Section~\ref{sec:obstructions}: the code has $d=1$, exactly
  as Proposition~\ref{prop:constant-distance} predicts for $p_0=1$.  This
  is the general shape of the difficulty, not an artifact of a toy code.
\end{example}

\subsection{A hypergraph-product instance}
\label{sec:hgp-example}

The first example shows that admissible triples need not be stable.  Let
$H_1$ be the $3\times7$ parity-check matrix of the $[7,4,3]$ Hamming code and
$H_2$ the $(\nu-1)\times\nu$ matrix with rows $e_i+e_{i+1}$, so
$\ker H_2=\{0^\nu,1^\nu\}$.  The hypergraph
product~\cite{tillich2014quantum} is
\begin{equation}
  H_X=[\,H_1\otimes I_\nu \mid I_3\otimes H_2^\top\,],\quad
  H_Z=[\,I_7\otimes H_2 \mid H_1^\top\otimes I_{\nu-1}\,],
\end{equation}
which satisfies
\begin{align}
  H_XH_Z^\top
  &=(H_1\otimes I_\nu)(I_7\otimes H_2^\top)
   +(I_3\otimes H_2^\top)(H_1\otimes I_{\nu-1})\notag\\
  &=H_1\otimes H_2^\top+H_1\otimes H_2^\top=0 .
\end{align}  We fix
\begin{equation}
  H_1=\begin{pmatrix}1&0&1&0&1&0&1\\0&1&1&0&0&1&1\\0&0&0&1&1&1&1\end{pmatrix},
  \label{eq:hamming}
\end{equation}
the parity-check matrix of the $[7,4,3]$ Hamming code in standard form; the
particular choice fixes which $\chi_i$ below are independent modulo $C_X$.

For $\nu=3$ this gives $n=27$ with $\operatorname{rank}H_X=9$,
$\operatorname{rank}H_Z=14$ and $k=4$; the row weights are $5$ to $6$ for
$H_X$ and $3$ to $5$ for $H_Z$, with column weights at most $3$ and $4$
respectively.  We computed $d_X$ by minimizing the weight over the $2^9$
elements of $C_X$ for each of the $15$ nonzero $X$-logical classes, and
$d_Z$ analogously over the $2^{14}$ elements of $C_Z$ for each nonzero
$Z$-logical class; both equal $3$, so the code is $[[27,4,3]]$.

Index the first sector of qubits by pairs $(i,c)\in[7]\times[\nu]$ and let
$\chi_i$ be the indicator of $\{i\}\times[\nu]$, a weight-$\nu$ vector
supported on one row of that sector.  Since $H_2\mathbf{1}^\nu=0$, each
$\chi_i$ lies in $\ker H_Z=C_Z^\perp$.  A rank computation with $H_1$ as in \eqref{eq:hamming} shows that
$\chi_1,\dots,\chi_4$ are linearly independent modulo $C_X$ and therefore
represent a basis of $\mathcal{L}_X$; for $\nu=3$ these are four weight-$3$
representatives on the disjoint blocks
$\{1,2,3\},\{4,5,6\},\{7,8,9\},\{10,11,12\}$.  This basis yields
representatives of much lower weight than the more familiar description of
the $X$-logicals through $\ker H_1$.

Form one canonical representative of each nonzero logical class by summing
the selected $\chi_i$.  Among the $\binom{15}{3}=455$ unordered triples of
these canonical representatives, $420$ have classes independent in
$\mathcal{L}_X$ and so are candidates.  Exactly four of the $420$ are
$\CCZ$-admissible, one for each way of omitting a single basis vector, and
each is saturated with $s_{xy}=s_{xz}=s_{yz}=s_{xyz}=3$, so $p_0=3$ and
$|S_t|=12$.  Lemma~\ref{lem:local-realization} synthesizes $111$ gates,
namely $81$ $\CCZ$, $27$ $\CZ$ and $3$ $Z$, with maximum qubit incidence
$37=1+3(3)+3(3)^2$; exhaustive evaluation over all $2^{12}$ basis states
of $S_t$ confirms \eqref{eq:wirewise-phase}.

Every part of the stability criterion fails for all four triples.  For
(S1), none of $x\odot y$, $x\odot z$, $y\odot z$ lies in $C_X^\perp$,
verified by computing $H_X(x\odot y)\ne0$ directly.  For (S2) and (S3), let $H_D$ be a basis matrix of the Schur span $D$,
which has $\dim D=24$ against $\dim C_X=9$.  The script verifies
$H_D\,v\ne0$ for each of $v\in\{x,y,z\}$, so (S2) fails, and $H_D\,g\ne0$
for every one of the nine generators $g$ of $C_X$, so (S3) fails and $C_X$
is not totally triorthogonal.  So a saturated triple on a finite qLDPC code can fail all three stability
conditions at once, which is why Proposition~\ref{prop:stability} must be
checked rather than assumed; the instance is consistent with
Proposition~\ref{prop:constant-distance}, since $d=3\le3p_0=9$.

For the four triples $Q=\ell_x\ell_y\ell_z$ is non-constant on the cosets of
$C_X$, so by Lemma~\ref{lem:codespace} none of the four circuits preserves
the code space and (H3) fails for each: the first triple has some $u\in C_X$
with $Q(u)=1\neq0=Q(0)$ and fails on all $16$ cosets, the other three on $8$
of $16$; the script reports both.  On this code the parity conditions are met
and the diagonal circuit they generate still does not act on the code
space.

Table~\ref{tab:hgp} reports the results of an exhaustive search over triples of chosen canonical representatives of $X$-logical classes, constructed via the hypergraph product of $\text{Hamming}(7,4)$ with the $(\nu-1)\times\nu$ path matrix for $\nu\in\{3,4,5\}$. Admissibility is representative-dependent, and the distance $d$ is exact for $\nu=3$ and an upper bound derived from the lowest-weight representative found for $\nu=4,5$. For $\nu=4$, none of the $420$ canonical triples is admissible: every region is a union of complete $\nu$-qubit blocks, so $|P_{xyz}|$ is a multiple of $\nu$, even when $\nu$ is, and $\tau=|P_{xyz}|\bmod2=0$. Admissibility of the canonical blockwise construction is thus sensitive to that parity, for these representatives.

\begin{table}[t]
  \centering
        \caption{Exhaustive search results for $\nu \in \{3,4,5\}$.}
  \label{tab:hgp}
  \small
    \begin{tabular}{@{}lcccccc@{}}
    \toprule
    $\nu$ & $\ker H_2$ & $[[n,k]]$ & $d$ & indep.\ triples & admissible & $p_0$ \\
    \midrule
    $3$ & odd  & $[[27,4]]$ & $3$      & 420 & 4 & 3 \\
    $4$ & even & $[[37,4]]$ & $\le4$   & 420 & 0 & --- \\
    $5$ & odd  & $[[47,4]]$ & $\le5$   & 420 & 4 & 5 \\
    \bottomrule
  \end{tabular}
\end{table}

\subsection{A code satisfying all the conditions}
\label{sec:positive}

The hypergraph-product instances fail stability, so it is natural to ask
whether the conditions are satisfiable at all.  They are.  On $n=12$ qubits,
index four disjoint triples of qubits as the regions
$P_{xy}=\{1,2,3\}$, $P_{xz}=\{4,5,6\}$, $P_{yz}=\{7,8,9\}$,
$P_{xyz}=\{10,11,12\}$, and take
\begin{align}
  x&=\mathbf 1_{P_{xy}\cup P_{xz}\cup P_{xyz}}, &
  y&=\mathbf 1_{P_{xy}\cup P_{yz}\cup P_{xyz}}, \notag\\
  z&=\mathbf 1_{P_{xz}\cup P_{yz}\cup P_{xyz}},
\end{align}
each of weight $9$, so the triple is saturated with $p=3$ and
$\tau(x,y,z)=|P_{xyz}|\bmod2=1$.  Let $C_X=\langle g\rangle$ with
\begin{equation}
  g=\mathbf 1_{\{1,2,4,5,7,8,10,11\}},
\end{equation}
a weight-$8$ vector meeting each region in exactly two qubits, and let $C_Z$
be the largest space orthogonal to $g,x,y,z$.  This is a genuine
$X$-stabilizer, not the trivial one: $C_X\neq\{0\}$.

Every condition holds, and the script verifies each.  The three pairwise
inner products vanish and $\tau=1$, so the triple is admissible.  For
stability, $g$ meets each region evenly, so $x\odot y,x\odot z,y\odot z$ are
orthogonal to $g$, giving (S1); $D=\langle g\rangle$ since $g\odot g=g$, and
$x,y,z$ are orthogonal to $g$, giving (S2); and $\wt(g)=8$ is even with
$\tau(g,g,g)=\wt(g)\bmod2=0$, giving (S3), so $C_X$ is totally
triorthogonal.  By the observation before Lemma~\ref{lem:codespace}, (S2)
already forces the synthesized circuit to preserve the code space, and the
coset-phase test confirms it does.  The synthesis produces $111$ gates with
maximum incidence $37=1+3(3)+3(3)^2$, matching
Corollary~\ref{cor:exact}, and reproduces \eqref{eq:wirewise-phase} on all
$2^{12}$ basis states.

The code parameters are $[[12,3,1]]$.  The distance is not an artifact of
this particular choice: it is forced, and the mechanism is the same
region-parity constraint that makes stability hold.

\begin{proposition}[Stability bounds the distance]
  \label{prop:pos-bound}
  Let $(x,y,z)$ be a saturated, stably $\CCZ$-admissible triple of a CSS
  code, with shared regions of sizes $a,b,c,e$ and $p=\max\{a,b,c,e\}$.
  Then every $g\in C_X$ meets each of the four regions in even size, and the
  code distance obeys $d\le 3p$.  Consequently no family of codes carrying
  such a triple with $p$ fixed can have $d\to\infty$: satisfying all the
  conditions confines the code to distance $O(p)$.
\end{proposition}

\begin{proof}
  Write $g_{xy}=|g\cap P_{xy}|$, and similarly for the other regions.  The
  Schur product $x\odot y$ is the indicator of $P_{xy}\cup P_{xyz}$, so
  $\inner{g}{x\odot y}=g_{xy}+g_{xyz}$; by (S1) this is $0$, and cyclically
  $g_{xz}+g_{xyz}=g_{yz}+g_{xyz}=0$, whence
  $g_{xy}\equiv g_{xz}\equiv g_{yz}\equiv g_{xyz}\pmod2$.  Saturation gives
  $\inner{x}{g}=g_{xy}+g_{xz}+g_{xyz}$, which is $0$ by (S2) since
  $x\in D^\perp\subseteq C_X^\perp$; substituting the common value $t$ gives
  $3t\equiv t\equiv0$, so all four intersections are even.  The weight bound
  $\wt(x)=a+b+e\le3p$ and nontriviality of the class of $x$ give $d\le3p$ as
  in Corollary~\ref{cor:exact}.  Fixing $p$ fixes the bound independently of
  $n$.
\end{proof}

This is the positive counterpart of Corollary~\ref{cor:exact}: the
conditions are satisfiable, but exactly in the bounded-distance regime the
corollary predicts.  A construction that keeps $\Delta$ bounded while
$d\to\infty$, the goal of a scalable native gate, must therefore violate one
of admissibility, saturation, stability, or the ancilla-free model, and the
constructions that succeed at scale do exactly that: Zhu's homological
product codes~\cite{zhu2025topological} and the addressable codes of He
\emph{et al.}~\cite{he2025quantum,he2025asymptotically} realize the logical
$\CCZ$ across three code blocks as a phase correct on the code space rather
than as an all-inputs identity on a single register, which is the
model-weakening escape named after Corollary~\ref{cor:exact}, not a
counterexample to Proposition~\ref{prop:pos-bound}.

\section{Discussion and Open Problems}
\label{sec:implications}

\subsection{Position relative to prior work}

Table~\ref{tab:unify} places Definition~\ref{def:admissible} next to three standard frameworks. The comparison is meant to show overlap, not equivalence. Those settings build in global structure that Definition~\ref{def:admissible} does not assume, so the statements here are correspondingly narrower. Triorthogonality constrains an entire generator matrix and reappears here, through (S3), as a stability requirement; the cup-product formulation works with cochains, where stability is intrinsic; and addressable orthogonality controls selected and spectator logical coordinates, which is the role played by (H3). Recent constructions provide the kind of global algebraic structure needed for logical actions of the type abstracted in (H3), and some also realize parallel disjoint logical gate patterns~\cite{guemard2025good,li2026transversal,nguyen2025good}, though the logical patterns differ from one construction to another. The contribution here is more focused: an exact stability criterion and its triorthogonality consequence, the seven-term saturated synthesis with uniqueness and the resulting bound $L\ge\Delta\ge3p\ge d$, two counting obstructions on candidate collections, and a negative computational example on a concrete code.

\begin{table}[t]
\centering
\caption{The parity condition inside established frameworks.}
\label{tab:unify}
\small
\begin{tabular}{@{}p{1.85cm}p{2.5cm}p{2.4cm}@{}}
\toprule
Framework & Overlap condition & Additional structure \\
\midrule
Triorthogonal \cite{bravyi2012magic} & even pairwise and triple overlaps among distinct rows & global constraints on all rows; odd/even logical--stabilizer split \\
Cup product \cite{bombin2006topological,zhu2025topological} & $[x]\cup[y]\cup[z]\ne0$ & cochain complex; stability intrinsic \\
Addressable orth.\ \cite{he2025quantum} & pairwise orth., $\tau=1$ on chosen logicals & control of spectator coordinates \\
This work & pairwise orth., $\tau=1$, fixed reps & none; stability characterized, (H3) assumed \\
\bottomrule
\end{tabular}
\end{table}

\subsection{From resource states to a production claim}

Parallel logical $\CCZ$ gates are not yet a resource-state construction. A production claim also needs encoded input, extraction and injection procedures, a noise model, and disjointness of the addressed logical subsystems, which is exactly (H4) rather than a consequence of disjoint physical supports. We therefore state a parallel gate implementation and leave the resource-state version open.

\subsection{Open problems}

The questions below all ask how far the bounded-congestion picture can be pushed, and what has to change once one leaves the strict ancilla-free setting.

\noindent\emph{1. Congestion at scale.}
Find an asymptotically good qLDPC family with $\omega(1)$ stably $\CCZ$-admissible triples whose diagonal realizations satisfy (H2) with constant $\Delta$. Proposition~\ref{prop:constant-distance} rules out the bounded-shared-region saturated route. By Corollary~\ref{cor:exact}, no saturated admissible triple can be realized at constant depth inside the ancilla-free $Z$-$\CZ$-$\CCZ$ model once the distance grows, regardless of region size. That leaves two directions: inside (H1)--(H4), one needs \emph{unsaturated} triples whose exact realizations still have bounded congestion; outside the present model, one may allow ancillas, non-diagonal or larger gates, or equality only on the code space. In that case Proposition~\ref{prop:scheduling} should be extended to the new circuit class.

\noindent\emph{2. Bounded congestion beyond $\mathcal{V}$.}
The cubic phase itself is not the hard part, and Lemma~\ref{lem:anf} shows that no ancilla-free $Z$-$\CZ$-$\CCZ$ circuit can improve on the monomial count. Progress therefore has to come from ancillas with non-diagonal gates, a larger gate set, or a phase that matches the target only on the code space. Note that $\Delta$, as defined here, counts edges of a rank-at-most-three hypergraph on data qubits and does not directly apply once ancilla vertices and higher-rank gates are allowed. The right question in that setting is whether the target phase admits a constant-depth bounded-arity circuit in which every data and ancilla qubit participates in $O(1)$ gates, or whether a depth or participation lower bound can be proved. Extending Proposition~\ref{prop:scheduling} to that circuit class is part of the problem.

\noindent\emph{3. Certifying (H3).}
Give coding-theoretic conditions, checkable from $H_X$ and $H_Z$, that guarantee code-space preservation and identify the induced logical phase as exactly the target $\CCZ$ monomial $a_i a_j a_\ell$, with no spectator dependence, for the circuit of Lemma~\ref{lem:local-realization}. Lemma~\ref{lem:codespace} settles code-space preservation for a fixed diagonal phase. What remains is to identify the induced Boolean phase function on $C_Z^\perp/C_X$ in a chosen logical basis and show that it is exactly the target $\CCZ$ monomial $a_i a_j a_\ell$, with no dependence on spectator logical coordinates. Stability (Proposition~\ref{prop:stability}) is necessary for the representative-level description to descend unchanged to classes, but it is not enough by itself, and the addressable-orthogonality conditions~\cite{he2025quantum} indicate the missing ingredients.

\noindent\emph{4. Stability-compatible families.}
Characterize CSS codes whose $C_X$ satisfies (S3), equivalently is totally triorthogonal in the sense of Proposition~\ref{prop:stability}, and that admit stably admissible triples. The hypergraph-product instances of Section~\ref{sec:hgp-example} fail (S1)--(S3) even though admissible representative triples exist, so the two properties do not come together automatically.

\noindent\emph{5. Higher arity.}
Extend the region decomposition and the seven-term identity to $r$-tuples and degree-$r$ multilinear overlap forms, where the analogue of \eqref{eq:seven-terms} governs $C^{r-1}Z$ phases.

\section*{Software Availability}
\label{app:repro}

The accompanying script~\cite{rowshan2026native} implements GF(2) row reduction, rank and nullspace;
the hypergraph-product construction with $H_XH_Z^\top=0$ asserted; a basis of
$C_Z^\perp/C_X$; enumeration of the $\binom{2^k-1}{3}$ class triples with a
rank filter; admissibility and the separate tests (S1)--(S3); the Schur span
$D$; exact minimum distance over both stabilizer sectors; the coset-phase
test of Lemma~\ref{lem:codespace}; and the incidence formulas
\eqref{eq:incidences} against the algebraic normal form.  Every
counterexample and count stated in the paper is asserted, and each quantity
in Section~\ref{sec:hgp-example} and Table~\ref{tab:hgp} is printed.
For efficiency $D$ is built from Schur products of a \emph{basis} of $C_X$
rather than all its elements, valid by bilinearity of $\odot$; and
Proposition~\ref{prop:stability}, proved analytically, is additionally
checked against direct enumeration over $C_X^3$ on $1000$ small instances.

\appendices

\section{The Seven-Term Identity}
\label{app:seven-term}

Expanding $(a+b+e)(a+c'+e)(b+c'+e)$ over $\Ftwo$ with $a,b,c',e$ idempotent
and reducing squares, the cross terms $ab$, $ac'$ and $abc'$ each appear an
even number of times and cancel, leaving Eq.~\eqref{eq:seven-terms}.  The
script checks the identity on all sixteen assignments of $(a,b,c',e)$.

\section{Light-Cone Bound}
\label{app:distance}

The cone must follow the circuit in order: an unrestricted closure under all
gates ignores time and can be far larger than $q^L|R|$, as a depth-two chain
of two-qubit gates already shows.

\begin{lemma}[Forward light cone]
  \label{lem:light-cone}
    Let the subcircuit remaining after the fault have layers
  $\mathcal{G}_1,\dots,\mathcal{G}_{L'}$, each gate acting on at most $q$
  qubits; the circuit need not be Clifford.  For $R\subseteq[n]$ put $R_0=R$
  and
  \begin{equation}
        R_j=R_{j-1}\cup\bigcup_{\substack{g\in\mathcal{G}_j\\
      \supp(g)\cap R_{j-1}\neq\emptyset}}\supp(g),
    \qquad j=1,\dots,L' ,
  \end{equation}
    and call $R_{L'}$ the \emph{forward light cone} of $R$.  Then
  $|R_j|\le q|R_{j-1}|$, so $|R_{L'}|\le q^{L'}|R|$.  Consequently a fault of
  weight $w$ inserted before layer $j$ of a depth-$L$ circuit propagates,
  through the $L'=L-j+1$ remaining layers, to an operator supported within a
  set of size at most $q^{L'}w\le q^Lw$.
\end{lemma}

\begin{proof}
  The gates of a single layer act on disjoint qubit sets, so each
  $i\in R_{j-1}$ lies in at most one $g\in\mathcal{G}_j$ and contributes at
  most $q$ elements (itself and the at most $q-1$ other legs) to $R_j$;
    hence $|R_j|\le q|R_{j-1}|$ and $|R_{L'}|\le q^{L'}|R|$ by induction.  The
  recursion records only which qubits a gate connects, so the bound is
  independent of what the gates do.
\end{proof}

Remark~\ref{rem:distance} records what this does and does not give.
Bounds of this type underlie the classification of locality-preserving
logical gates~\cite{bravyi2013classification,pastawski2015fault}.

\bibliographystyle{IEEEtran}
\bibliography{refs1}

@inproceedings{aharonov1997fault,
  author    = {Aharonov, Dorit and Ben-Or, Michael},
  title     = {Fault-tolerant quantum computation with constant error},
  booktitle = {Proc. 29th ACM Symp. Theory of Computing (STOC)},
  pages     = {176--188},
  year      = {1997}
}

@article{aliferis2006quantum,
  author  = {Aliferis, P. and Gottesman, D. and Preskill, J.},
  title   = {Quantum accuracy threshold for concatenated distance-3 codes},
  journal = {Quantum Inf. Comput.},
  volume  = {6},
  number  = {2},
  pages   = {97--165},
  year    = {2006}
}

@article{kitaev2003fault,
  author  = {Kitaev, A. Yu.},
  title   = {Fault-tolerant quantum computation by anyons},
  journal = {Ann. Phys.},
  volume  = {303},
  number  = {1},
  pages   = {2--30},
  year    = {2003}
}

@article{bombin2006topological,
  author  = {Bomb{\'i}n, H. and Martin-Delgado, M. A.},
  title   = {Topological quantum distillation},
  journal = {Phys. Rev. Lett.},
  volume  = {97},
  pages   = {180501},
  year    = {2006}
}

@article{kubica2015universal,
  author  = {Kubica, A. and Beverland, M. E.},
  title   = {Universal transversal gates with color codes: A simplified approach},
  journal = {Phys. Rev. A},
  volume  = {91},
  pages   = {032330},
  year    = {2015}
}

@article{vasmer2019three,
  author  = {Vasmer, M. and Browne, D. E.},
  title   = {Three-dimensional surface codes: Transversal gates and fault-tolerant architectures},
  journal = {Phys. Rev. A},
  volume  = {100},
  pages   = {012312},
  year    = {2019}
}

@article{eastin2009restrictions,
  author  = {Eastin, B. and Knill, E.},
  title   = {Restrictions on transversal encoded quantum gate sets},
  journal = {Phys. Rev. Lett.},
  volume  = {102},
  pages   = {110502},
  year    = {2009}
}

@article{bravyi2005universal,
  author  = {Bravyi, S. and Kitaev, A.},
  title   = {Universal quantum computation with ideal {C}lifford gates and noisy ancillas},
  journal = {Phys. Rev. A},
  volume  = {71},
  pages   = {022316},
  year    = {2005}
}

@article{bravyi2012magic,
  author  = {Bravyi, S. and Haah, J.},
  title   = {Magic-state distillation with low overhead},
  journal = {Phys. Rev. A},
  volume  = {86},
  pages   = {052329},
  year    = {2012}
}

@article{campbell2017quantum,
  author  = {Campbell, E. T. and O'Gorman, J.},
  title   = {Quantum computation with realistic magic-state factories},
  journal = {Phys. Rev. A},
  volume  = {95},
  pages   = {032338},
  year    = {2017}
}

@article{gidney2019efficient,
  author  = {Gidney, C. and Fowler, A. G.},
  title   = {Efficient magic state factories with a catalyzed {CCZ} to {2T} transformation},
  journal = {Quantum},
  volume  = {3},
  pages   = {135},
  year    = {2019}
}

@article{litinski2019game,
  author  = {Litinski, D.},
  title   = {A game of surface codes: Large-scale quantum computing with lattice surgery},
  journal = {Quantum},
  volume  = {3},
  pages   = {128},
  year    = {2019}
}

@article{breuckmann2021quantum,
  author  = {Breuckmann, N. P. and Eberhardt, J. N.},
  title   = {Quantum low-density parity-check codes},
  journal = {PRX Quantum},
  volume  = {2},
  pages   = {040101},
  year    = {2021}
}

@article{tillich2014quantum,
  author  = {Tillich, J.-P. and Z{\'e}mor, G.},
  title   = {Quantum {LDPC} codes with positive rate and minimum distance proportional to the square root of the blocklength},
  journal = {IEEE Trans. Inf. Theory},
  volume  = {60},
  number  = {2},
  pages   = {1193--1202},
  year    = {2014}
}

@inproceedings{panteleev2021asymptotically,
  author    = {Panteleev, Pavel and Kalachev, Gleb},
  title     = {Asymptotically good quantum and locally testable classical {LDPC} codes},
  booktitle = {Proc. 54th Annual ACM SIGACT Symp. Theory of Computing (STOC)},
  pages     = {375--388},
  year      = {2022},
  doi       = {10.1145/3519935.3520017},
  eprint    = {2111.03654},
  archivePrefix = {arXiv},
  primaryClass  = {cs.IT}
}

@inproceedings{leverrier2022quantum,
  author    = {Leverrier, A. and Z{\'e}mor, G.},
  title     = {Quantum {T}anner codes},
  booktitle = {Proc. 63rd IEEE Symp. Foundations of Computer Science (FOCS)},
  pages     = {872--883},
  year      = {2022}
}

@article{bravyi2010tradeoffs,
  author  = {Bravyi, S. and Poulin, D. and Terhal, B.},
  title   = {Tradeoffs for reliable quantum information storage in {2D} systems},
  journal = {Phys. Rev. Lett.},
  volume  = {104},
  pages   = {050503},
  year    = {2010}
}

@article{baspin2022quantifying,
  author  = {Baspin, N. and Krishna, A.},
  title   = {Quantifying nonlocality: How outperforming local quantum codes is expensive},
  journal = {Phys. Rev. Lett.},
  volume  = {129},
  pages   = {050505},
  year    = {2022}
}

@article{tremblay2022constant,
  author  = {Tremblay, M. A. and Delfosse, N. and Beverland, M. E.},
  title   = {Constant-overhead quantum error correction with thin planar connectivity},
  journal = {Phys. Rev. Lett.},
  volume  = {129},
  pages   = {050504},
  year    = {2022}
}

@article{bravyi2024high,
  author  = {Bravyi, S. and Cross, A. W. and Gambetta, J. M. and Maslov, D. and Rall, P. and Yoder, T. J.},
  title   = {High-threshold and low-overhead fault-tolerant quantum memory},
  journal = {Nature},
  volume  = {627},
  pages   = {778--782},
  year    = {2024}
}

@article{pecorari2025highrate,
  author  = {Pecorari, L. and Jandura, S. and Brennen, G. K. and Pupillo, G.},
  title   = {High-rate quantum {LDPC} codes for long-range-connected neutral atom registers},
  journal = {Nat. Commun.},
  volume  = {16},
  pages   = {1111},
  year    = {2025}
}

@article{xu2025fast,
  author  = {Xu, Qian and Zhou, Hengyun and Zheng, Guo and Bluvstein, Dolev
             and Bonilla Ataides, J. Pablo and Lukin, Mikhail D. and Jiang, Liang},
  title   = {Fast and parallelizable logical computation with homological product codes},
  journal = {Phys. Rev. X},
  volume  = {15},
  pages   = {021065},
  year    = {2025},
  doi     = {10.1103/PhysRevX.15.021065}
}

@inproceedings{golowich2024asymptotically,
  author    = {Golowich, Louis and Guruswami, Venkatesan},
  title     = {Asymptotically good quantum codes with transversal non-{C}lifford gates},
  booktitle = {Proc. 57th Annual ACM Symp. Theory of Computing (STOC)},
  pages     = {707--717},
  year      = {2025},
  eprint    = {2408.09254},
  archivePrefix = {arXiv},
  primaryClass  = {quant-ph}
}

@inproceedings{golowich2024quantum,
  author    = {Golowich, Louis and Lin, Ting-Chun},
  title     = {Quantum {LDPC} codes with transversal non-{C}lifford gates via
               products of algebraic codes},
  booktitle = {Proc. 57th Annual ACM Symp. Theory of Computing (STOC)},
  pages     = {689--696},
  year      = {2025},
  eprint    = {2410.14662},
  archivePrefix = {arXiv},
  primaryClass  = {quant-ph}
}

@misc{lin2024transversal,
  author       = {Lin, Ting-Chun},
  title        = {Transversal non-{C}lifford gates for quantum {LDPC} codes on sheaves},
  year         = {2024},
  eprint       = {2410.14631},
  archivePrefix = {arXiv},
  primaryClass = {quant-ph},
  note         = {arXiv:2410.14631}
}

@misc{he2025quantum,
  author       = {He, Zhiyang and Vaikuntanathan, Vinod and Wills, Adam and Zhang, Rachel Yun},
  title        = {Quantum codes with addressable and transversal non-{C}lifford gates},
  year         = {2025},
  eprint       = {2502.01864},
  archivePrefix = {arXiv},
  primaryClass = {quant-ph},
  note         = {arXiv:2502.01864}
}

@misc{he2025asymptotically,
  author       = {He, Zhiyang and Vaikuntanathan, Vinod and Wills, Adam and Zhang, Rachel Yun},
  title        = {Asymptotically good quantum codes with addressable and transversal non-{C}lifford gates},
  year         = {2025},
  eprint       = {2507.05392},
  archivePrefix = {arXiv},
  primaryClass = {quant-ph},
  note         = {arXiv:2507.05392}
}

@misc{zhu2025topological,
  author       = {Zhu, Guanyu},
  title        = {A topological theory for {qLDPC}: Non-{C}lifford gates and magic state
                  fountain on homological product codes with constant rate and beyond the
                  $N^{1/3}$ distance barrier},
  year         = {2025},
  eprint       = {2501.19375},
  archivePrefix = {arXiv},
  primaryClass = {quant-ph},
  note         = {arXiv:2501.19375}
}

@misc{guemard2025good,
  author       = {Gu{\'e}mard, Virgile},
  title        = {Good quantum codes with addressable and parallelizable transversal
                  non-{C}lifford gates},
  year         = {2025},
    eprint       = {2510.19809},
  archivePrefix = {arXiv},
  primaryClass = {quant-ph},
  note         = {arXiv:2510.19809v2}
}

@misc{li2026transversal,
  author       = {Li, Yiming and Li, Zimu and Liu, Zi-Wen},
  title        = {Transversal non-{C}lifford gates on almost-good quantum {LDPC} and
                  quantum locally testable codes},
  year         = {2026},
  eprint       = {2604.01874},
  archivePrefix = {arXiv},
  primaryClass = {quant-ph},
  note         = {arXiv:2604.01874}
}

@inproceedings{nguyen2025good,
  author    = {Nguyen, Quynh T.},
  title     = {Good binary quantum codes with transversal {CCZ} gate},
  booktitle = {Proc. 57th Annual ACM Symp. Theory of Computing (STOC)},
  pages     = {697--706},
  year      = {2025},
  eprint    = {2408.10140},
  archivePrefix = {arXiv},
  primaryClass  = {quant-ph}
}

@misc{fu2025nogo,
  author       = {Fu, Esther Xiaozhen and Zheng, Han and Li, Zimu and Liu, Zi-Wen},
  title        = {No-go theorems for logical gates on product quantum codes},
  year         = {2025},
  eprint       = {2507.16797},
  archivePrefix = {arXiv},
  primaryClass = {quant-ph},
  note         = {arXiv:2507.16797}
}

@article{bravyi2013classification,
  author  = {Bravyi, S. and K{\"o}nig, R.},
  title   = {Classification of topologically protected gates for local stabilizer codes},
  journal = {Phys. Rev. Lett.},
  volume  = {110},
  pages   = {170503},
  year    = {2013}
}

@article{pastawski2015fault,
  author  = {Pastawski, F. and Yoshida, B.},
  title   = {Fault-tolerant logical gates in quantum error-correcting codes},
  journal = {Phys. Rev. A},
  volume  = {91},
  pages   = {012305},
  year    = {2015}
}

@article{pippenger1989asymptotic,
  author  = {Pippenger, N. and Spencer, J.},
  title   = {Asymptotic behavior of the chromatic index for hypergraphs},
  journal = {J. Combin. Theory Ser. A},
  volume  = {51},
  number  = {1},
  pages   = {24--42},
  year    = {1989}
}

@misc{rowshan2026native,
  author       = {Mohammad Rowshan},
  title        = {Native Non-Clifford Gates in qLDPC Codes},
  year         = {2026},
  publisher    = {GitHub},
  journal      = {GitHub repository},
  howpublished = {\url{https://github.com/mohammad-rowshan/Native-Non-Clifford-Gates-in-qLDPC-Codes}}
}

\end{document}